\documentclass[11pt]{article}
\usepackage[utf8]{inputenc}
\usepackage[margin=1in]{geometry}
\usepackage{enumitem}
\usepackage{amsmath, amssymb, amsthm}
\usepackage{mathtools}
\usepackage{tikz, tikz-cd}
\usetikzlibrary{decorations.markings}
\usepackage{xcolor}
\usepackage{verbatim}
\usepackage{graphicx}
\usepackage{hyperref}
\usepackage{titlesec}
\usepackage{braket}
\usepackage{yhmath}
\usepackage{indentfirst}
\usepackage{bm}
\usepackage{natbib}
\usepackage{authblk}

\makeatletter
\let\save@mathaccent\mathaccent
\newcommand*\if@single[3]{%
  \setbox0\hbox{${\mathaccent"0362{#1}}^H$}%
  \setbox2\hbox{${\mathaccent"0362{\kern0pt#1}}^H$}%
  \ifdim\ht0=\ht2 #3\else #2\fi
  }
\newcommand*\rel@kern[1]{\kern#1\dimexpr\macc@kerna}
\newcommand*\widebar[1]{\@ifnextchar^{{\wide@bar{#1}{0}}}{\wide@bar{#1}{1}}}
\newcommand*\wide@bar[2]{\if@single{#1}{\wide@bar@{#1}{#2}{1}}{\wide@bar@{#1}{#2}{2}}}
\newcommand*\wide@bar@[3]{%
  \begingroup
  \def\mathaccent##1##2{%
    \let\mathaccent\save@mathaccent
    \if#32 \let\macc@nucleus\first@char \fi
    \setbox\z@\hbox{$\macc@style{\macc@nucleus}_{}$}%
    \setbox\tw@\hbox{$\macc@style{\macc@nucleus}{}_{}$}%
    \dimen@\wd\tw@
    \advance\dimen@-\wd\z@
    \divide\dimen@ 3
    \@tempdima\wd\tw@
    \advance\@tempdima-\scriptspace
    \divide\@tempdima 10
    \advance\dimen@-\@tempdima
    \ifdim\dimen@>\z@ \dimen@0pt\fi
    \rel@kern{0.6}\kern-\dimen@
    \if#31
      \overline{\rel@kern{-0.6}\kern\dimen@\macc@nucleus\rel@kern{0.4}\kern\dimen@}%
      \advance\dimen@0.4\dimexpr\macc@kerna
      \let\final@kern#2%
      \ifdim\dimen@<\z@ \let\final@kern1\fi
      \if\final@kern1 \kern-\dimen@\fi
    \else
      \overline{\rel@kern{-0.6}\kern\dimen@#1}%
    \fi
  }%
  \macc@depth\@ne
  \let\math@bgroup\@empty \let\math@egroup\macc@set@skewchar
  \mathsurround\z@ \frozen@everymath{\mathgroup\macc@group\relax}%
  \macc@set@skewchar\relax
  \let\mathaccentV\macc@nested@a
  \if#31
    \macc@nested@a\relax111{#1}%
  \else
    \def\gobble@till@marker##1\endmarker{}%
    \futurelet\first@char\gobble@till@marker#1\endmarker
    \ifcat\noexpand\first@char A\else
      \def\first@char{}%
    \fi
    \macc@nested@a\relax111{\first@char}%
  \fi
  \endgroup
}
\makeatother

\theoremstyle{plain}
\newtheorem{theorem}{Theorem}[section]

\newtheorem{lemma}[theorem]{Lemma}

\theoremstyle{definition}
\newtheorem{definition}[theorem]{Definition}
\newtheorem{assumption}{Assumption}

\newtheorem{problem}{Problem}
\newenvironment{problem*}[1]
  {\renewcommand\theproblem{#1}\problem}
  {\endproblem}

\theoremstyle{remark}

\numberwithin{equation}{section}

\newcommand{\mat}[1]{\bm{#1}}

\newcommand{\C}{\mathbb{C}}
\newcommand{\N}{\mathbb{N}}

\newcommand{\Z}{\mathbb{Z}}

\newcommand\norm[1]{\left\lVert#1\right\rVert}

\titleformat{\subsection}[runin]{\bfseries}{}{}{}

\begin{document}

\title{Edge Hamiltonian for Free Fermion Quantum Hall Models}

\author[1]{Simon Du}
\author[1]{Martin Fraas}
\author[1]{Abi Gopal}
\author[1]{Nathan Singh}

\affil[1]{Department of Mathematics, University of California, Davis}

\date{\today }

\maketitle

\begin{abstract} 
We investigate a proposal of Kitaev for a microscopic construction of a Hamiltonian intended to describe the edge dynamics of a quantum Hall system. We show that the construction works in the setting of translation-invariant free-fermion Hamiltonians. In this case, the resulting edge Hamiltonian exhibits only edge modes, and these modes have the correct chirality.

\end{abstract}

\section{Introduction}

This article is the authors’ exploration of a proposal of Kitaev \cite{Kitaev} for defining and computing the chiral central charge of microscopic models.

The standard theoretical physics definition \cite{Kane_Fisher_97} of the chiral central charge, $c$, is formulated in terms of edge currents. On the edge of a material that is gapped in the bulk, there may flow an energy current $J$. For temperatures $T$ much smaller than the bulk gap,
\begin{equation}
\label{eq:1}
J = \frac{\pi}{12}\, T^2\, c.
\end{equation}
Here we use units where $k_b = 1$ and $h = 2\pi$. The prefactor $\pi/12$ is chosen so that, for free fermions, the chiral central charge agrees with the Hall conductance and, in particular, is integer valued. The $T^2$ dependence follows from scaling symmetry \cite{Kapustin_2019}.

Let us briefly recall the part of \cite{Kane_Fisher_97} that is relevant to the motivation for this work. Let $H$ be a translation-invariant free-fermion Hamiltonian on the square lattice, and let $H_{\mathrm{edge}}$ denote the Hamiltonian of a system with an edge along the line $x = 0$. The edge Hamiltonian remains translation invariant in the $y$-direction and thus has a fiber decomposition $H_{\mathrm{edge}}(k_y)$.
The spectrum of $H_{\mathrm{edge}}$ consists of the bulk spectrum of $H$ together with the edge spectrum. The latter consists of finitely many edge modes
$ E_1(k_y),\dots, E_n(k_y)$ that appear in the bulk gaps. An edge mode $E(k)$ is called chiral if it connects the band below the gap with the band above the gap. Let $k_-$ and $k_+$ be the momenta at which the mode merges with the lower and upper bands, respectively. Its chirality is defined as
\[
\mathrm{sign}(E) =
\begin{cases}
+1, & k_- < k_+, \\
-1, & k_- > k_+, \\
0, & \text{if the mode does not connect the two bands}.
\end{cases}
\]
It is well known \cite{Hatsugai, Graf} that the Hall conductance is given by the sum
$\mathrm{sign}(E_1) + \cdots + \mathrm{sign}(E_n)$.

The energy-current observable is formally
\[
H_{\mathrm{edge}}(k_y)\, \frac{\partial H_{\mathrm{edge}}(k_y)}{\partial k_y}.
\]
Although this operator is not trace class, one may formally write the energy current as
\begin{equation}
\label{eq:J}
J = \int_0^{2\pi} \mathrm{tr}\!\left( 
H_{\mathrm{edge}}(k_y) 
\frac{\partial H_{\mathrm{edge}}(k_y)}{\partial k_y} 
\frac{1}{e^{\beta H_{\mathrm{edge}}(k_y)} + 1}
- P
\right)\,
\frac{dk_y}{2\pi},
\end{equation}
where $P$ is the Fermi projection of the ground state. The chemical potential is set to $0$. Restricting the trace to the edge spectrum yields a sum of integrals of the form
\[
\int_{k_-}^{k_+} 
E(k)\, E'(k)\,
\Bigl(\frac{1}{e^{\beta E(k)} + 1} - \chi(-E(k)) \Bigr)
\frac{dk}{2\pi},
\]
where $\chi$ is the Heaviside step function. To leading order in $T$ one obtains
\[
\mathrm{sign}(E)\, T^2 \cdot 2\! \int_0^\infty \frac{x}{e^x+1}\,\frac{dx}{2\pi}
= \mathrm{sign}(E) \, T^2\, \frac{\pi}{12}.
\]
Thus
\[
J = \frac{\pi}{12} T^2 \sum_{j=1}^n \mathrm{sign}(E_j),
\]
showing that $c$ equals the Hall conductance.

The key point we want to emphasize is that obtaining the correct result relies on restricting the trace to the edge modes and effectively ignoring the bulk. Kane and Fisher make the same choice beyond the free-fermion case: they work with an effective chiral Luttinger liquid model for the edge modes and compute the current within that model, without including bulk contributions. A question how to extract an effective edge dynamics from a microscopic model remains open.

\vspace{0.3cm}

\noindent\textbf{Kitaev’s proposal.}
In Appendix~D of \cite{Kitaev}, two related ideas are introduced. First, Kitaev proposes a microscopic construction of a Hamiltonian, which we denote by $\widetilde{H}_{\mathrm{edge}}$, intended to model the edge dynamics only. Second, he proposes a way to compute the chiral central charge using only bulk quantities. The second proposal was later developed in detail by Kapustin and Spodyneiko \cite{Kapustin2020}, who connected it to linear response via the Kubo formula and proved that for free fermions it agrees with the Hall conductance. To the best of our knowledge, the first proposal has not yet been explored. In this article we study it for translation-invariant free fermions and show that it indeed produces an edge-only Hamiltonian with the correct number of chiral modes.

Kitaev considers a many-body Hamiltonian
\[
H = \sum_{x,y} H_{x,y},
\]
decomposed into local terms $H_{x,y}$ anchored\footnote{For the notion of anchored interactions, see \cite{bachmann2024classification}.} at the point $(x,y)$. He then proposes the edge Hamiltonian
\[
\widetilde{H}_{\mathrm{edge}}
= \sum_{x,y} g(x)\, H_{x,y},
\]
where $g(x)$ is a monotonically increasing function satisfying
\[
\lim_{x\to -\infty} g(x) = 0, \qquad
\lim_{x\to +\infty} g(x) = +\infty,
\]
and which should grow sub-exponentially. The precise shape of $g$ is likely unimportant, but it must not grow too fast. The heuristic is that for $x \ll 0$, the contribution of $H_{x,y}$ is energetically suppressed so that, at any positive temperature, the state is close to the maximally mixed state. On the other hand, when restricting to an interval $x\in (L,L+\ell)$ with $L\gg 0$,
\[
\widetilde{H}_{\mathrm{edge}}
\approx g(L) \sum_{L<x<L+\ell} 
\frac{g(x)}{g(L)}\, H_{x,y}.
\]
For $L\gg 0$ the ratio $g(x)/g(L)$ is approximately $1$ on this interval. Thus the Hamiltonian there is a large multiple of $H$ and the state is close to the bulk ground state if $\ell$ is much larger than the correlation length. So the bulk is effectively sent to infinite energy. In particular, any finite-energy state has support only at finite $x$.

In this article we investigate this construction for free fermions with the explicit choice
\[
g(x)=
\begin{cases}
0, & x<0,\\
x, & x \geq 0.
\end{cases}
\]
We show that the single-particle Hamiltonian $\widetilde{H}_{\mathrm{edge}}(k_y)$ has discrete spectrum—as an edge Hamiltonian should—and that its spectral flow equals the Hall conductance. It follows that (\ref{eq:J}), with $\widetilde{H}_{\mathrm{edge}}(k_y)$ in place of $H_{\mathrm{edge}}(k_y)$, is well defined and reproduces (\ref{eq:1}).

Finally, before describing the results in detail, we note that deriving edge dynamics from a microscopic Hamiltonian is of independent interest beyond the study of chiral central charge. In field-theoretic approaches, the boundary action plays a central role in the theory of the fractional quantum Hall effect \cite{Frohlich1, Frohlich2, Frohlich3}. Microscopically we have detailed understanding of bulk-edge correspondence (see \cite{Bols} for the most recent developments). Edge dynamics have been described using semiclassical methods \cite{drouot2021, Drouot}.  For the interacting case \cite{Porta, Porta2} proved that correlations along the edge agree with the Luttinger liquid model. However as far as we know there is no microscopic construction of a Hamiltonian that would describe edge only dynamics. 

\section{Setup and Results}
\label{sec:setup}
We consider a lattice $\mathbb{Z}^2$ with the usual notation $x, y$ for the two axes. We equip $\mathbb{Z}^2$ with the $\ell^\infty$ distance. Let $H$ be a finite range one-periodic Hamiltonian on $\ell^2(\Z^2, \mathbb{C}^n)$. After suitable coarse graining we can without loss of generality assume that the range is $1$, meaning that there is no hopping between any sites with distance larger than $1$. The Fourier transform of $H$ acts as a multiplication operator on $L^2([0, 2\pi)^2,\, \mathbb{C}^n)$, and we denote it by $H(k_x, k_y)$. 
Since $H$ is range 1, we may write $H(k_x, k_y) = V(k_y) + A(k_y) e^{i k_x} + A^*(k_y) e^{-i k_x}$ where $A,A^*,V$ are $n \times n$ matrices and $V$ is Hermitian.

We assume that $H(k_x, k_y)$ has $n$ gapped bands.
\begin{assumption}
\label{assumption:bulk}
The $n \times n$ matrix $H(k_x, k_y)$ has $n$ non-degenerate eigenvalues $E_1(k_x, k_y) < E_2(k_x, k_y) < \dots < E_n(k_x, k_y)$, and for any $ j =1, \dots, n-1$, there exists number $E_F$, such that 
$$
E_j(k_x, k_y) < E_F < E_{j+1}(k_x, k_y),
$$
for all $k_x, k_y$. 

In addition, we assume that there exists an integer $k$, $1 \leq k < n$, such that $E_F = 0$ is in the gap between $(k-1)$th and $k$th bands, i.e. $E_{k-1}(k_x, k_y) < 0 < E_{k}(k_x, k_y).$
\end{assumption}
We do not think that having $n$ gapped bands is essential, and the bands may merge.  The reason for the choice is that we are trying to demonstrate Kitaev's proposal in the simplest possible setting. The last part of the assumption, that zero lies in the gap, is set for convenience; we can always obtain that by shifting the energy.  In line with the assumption, we will set the Fermi energy to be $E_F = 0$. The Hall conductance, $\kappa$, of the ground state is then given by the sum of Chern numbers, $c_j$, of the bands $j = 1, \dots, k -1$.

Let $\varphi_j(k_x, k_y)$ be eigenstates of $H(k_x, k_y)$ corresponding to the energy $E_j(k_x, k_y)$. We pick $\varphi_j(k_x, k_y)$ smooth in $[0, 2\pi) \times [0, 2\pi)$ with 
$$
\varphi_j(2 \pi, k_y) = e^{i \theta_j(k_y)} \varphi_j(0, k_y), \quad \varphi_j(k_x, 2 \pi) = \varphi_j(k_x, 0),
$$
where $\theta_j(k)$ is called the Berry phase. It is well known that the Chern number $c_j$ is equal to the winding number of $\theta_j(k)$, i.e., $c_j = \frac{1}{2\pi} (\theta_j(2\pi) - \theta_j(0))$.

We decompose $H = \sum_{x} H_x$ for $H_x$ self adjoint, where $H_x$ consists of all on-site terms at $x$ and all hopping terms between site $x$ and site $x + 1$, for all $y$. We write $H(k_y) = \sum_{x} H_x(k_y)$ for the partial Fourier transform of $H$ in the $y$-direction.
\begin{definition}\label{Htildedef}
We define the distance modulated Hamiltonian 
  \begin{gather}
  \widetilde{H} \coloneq \sum_{x} x H_x.
  \end{gather}
We define $\widetilde{H}_{\mathrm{edge}}$ as the restriction of $\widetilde{H}$ to $\ell^2(\mathbb{N} \times \mathbb{Z}, \mathbb{C}^n)$. The restriction is well defined by our choice of $H_1$,  which has no hopping between $x=1$ and $x=0$.
\end{definition}
To be precise, since $\widetilde{H}$ is an unbounded operator, the definition domain of $\widetilde{H}$ is all  $\psi \in \ell^2(\mathbb{Z}^2, \mathbb{C}^n)$ for which $\sum_{x,y} \|\psi(x,y)\|_{\mathbb{C}^n}^2 x^2 < \infty$. With this definition domain, Assumption~\ref{assumption:bulk} implies that $\widetilde{H}$ is self-adjoint -- this will be clear once we compute the partial Fourier transform of $\widetilde{H}(k_y)$. We denote the partial Fourier transform of $\widetilde{H}_{\mathrm{edge}}$ by $\widetilde{H}_{\mathrm{edge}}(k_y)$. The operator $\widetilde{H}_{\mathrm{edge}}(k_y)$ depends analytically  on $k_y$ and $\widetilde{H}_{\mathrm{edge}}(0) = \widetilde{H}_{\mathrm{edge}}(2 \pi)$. For a continuous family of self-adjoint operators defined on a circle with discrete spectrum, the spectral flow is defined as the number of signed crossings of eigen-branches through  a generic fiducial line \cite{Phillips}.

Our main result is
\begin{theorem}\label{WindingNumber}
	The spectrum of $\widetilde{H}_{\mathrm{edge}}(k_y)$ is discrete. 
The spectral flow of the spectrum of $\widetilde{H}_{\text{edge}}(k_y)$
  is equal to the Hall conductance.  
\end{theorem}
The proof is based on an high-energy estimates for the spectrum. For large positive energies the spectrum of $\widetilde{H}(k_y)$ is close to the set of energies
 \begin{gather}\label{approxSpectrum}
	 E(m,j, k_y) = \frac{2 \pi m + \theta_j(2\pi, k_y) + \int_{0}^{2\pi} E_j(k_x, k_y)^{-1} e^{-i k_x} \braket{\psi_j(k_x, k_y) | A^* | \psi_j(k_x, k_y)} dk_x}{\int_{0}^{2\pi} E_j(k_x, k_y)^{-1} dk_x}
  \end{gather}
  with  $m \in \Z$, and  $1 \leq j \leq n$. We write $\sigma_{approx}(k_y) := \{ E(m,j, k_y), m \in \Z,\, 1 \leq j \leq n$ \}.
\begin{theorem}\label{Formula}
There exists a constant $C$ such that for any positive $E \in \sigma(\widetilde{H}(k_y))$,
$$
\mathrm{dist} (E, \sigma_{approx}(k_y)) \leq C E^{-1},
$$
and for any positive $E \in \sigma_{approx}(k_y)$,
$$
\mathrm{dist} (E, \sigma(\widetilde{H}(k_y))) \leq C E^{-1}.
$$

The large positive eigenvalues of $\widetilde{H}_{\text{edge}}(k_y)$ are close, in the above sense, to  
  $E(m,j, k_y)$ for $m > 0$, $1 \leq j \leq k-1$.
\end{theorem}
As an example, consider the Harper-Hofstadter model,
where the Hamiltonian is given by 
\begin{gather}
  (H(k_y) \psi)_m = \psi_{m-1} + \psi_{m+1} + 2 \cos(2\pi m \alpha - k_y)
\end{gather}
for $\alpha$ a free parameter. 
We take $\alpha = 1/3$. Then the Fourier transform of $H$ looks like 
\begin{gather*}
  H(k_x, k_y) = \begin{pmatrix} 2 \cos\left(\frac{4 \pi}{3} - k_y \right) & 1 & e^{-i k_x} \\
    1 & 2 \cos(k_y) & 1 \\
    e^{ik_x} & 1 & 2 \cos\left( \frac{2 \pi}{3} - k_y \right) \end{pmatrix} - E_F
    = V + A e^{i k_x} + A^* e^{-i k_x}
\end{gather*}
where 
\begin{gather*}
  V = \begin{pmatrix} 2 \cos\left(\frac{4 \pi}{3} - k_y \right) - E_F & 1 & 0 \\
    1 & 2 \cos(k_y) - E_F & 1 \\
    0 & 1 & 2 \cos\left( \frac{2 \pi}{3} - k_y \right) - E_F \end{pmatrix} \\
  A = \begin{pmatrix}
    0 & 0 & 0 \\ 0 & 0 & 0 \\ 1 & 0 & 0
  \end{pmatrix}.
\end{gather*}
There are two gaps seen in Figure~\ref{Harper}; Fermi energies $E_F = \pm 1.5$ lie in these gaps. 
\begin{figure}[h!]
  \centering
  \includegraphics[totalheight=4cm]{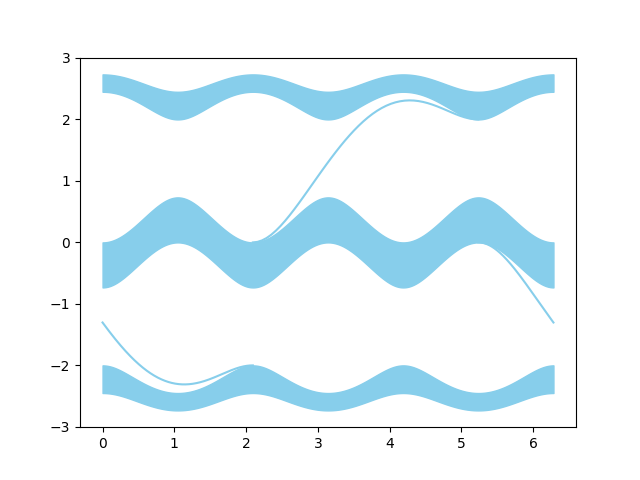}
  \caption{The spectrum of the Harper-Hofstadter model on a half plane with $\alpha = 1/3$. The model has three bands. The lines connecting the bands are the edge modes. The Chern numbers of these bands are $-1, 2, -1$.}
  \label{Harper}
\end{figure}
In Figure~\ref{Approx Spectrum} we demonstrate how \eqref{approxSpectrum} depends on  which bands are below the Fermi energy. In our theorems, we set $E_F=0$ for convenience, translating a given Fermi energy to that setting means considering the energy shift $H \to H - E_F$.
\begin{figure}[h!]
    \centering
  \includegraphics[totalheight=4cm]{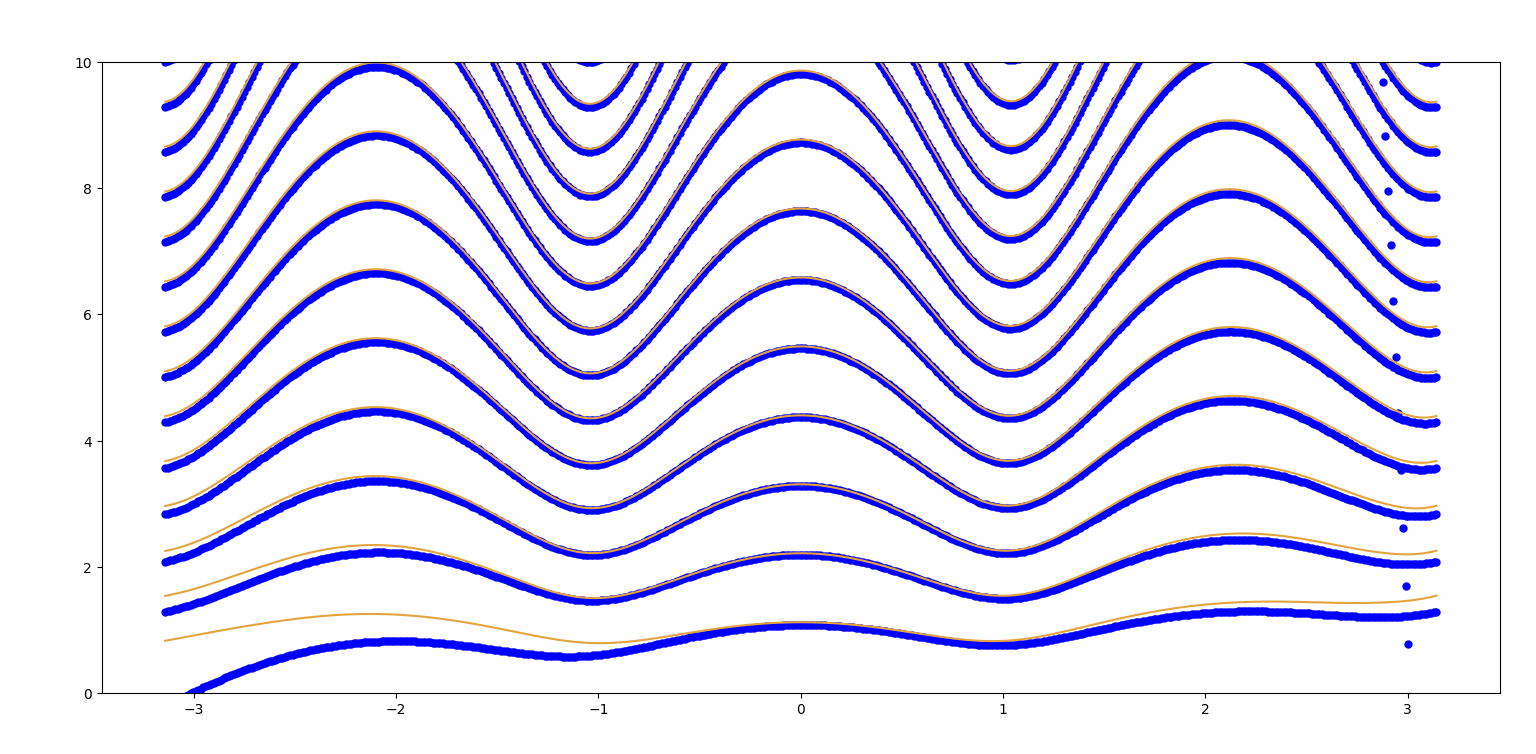}
  \includegraphics[totalheight=4cm]{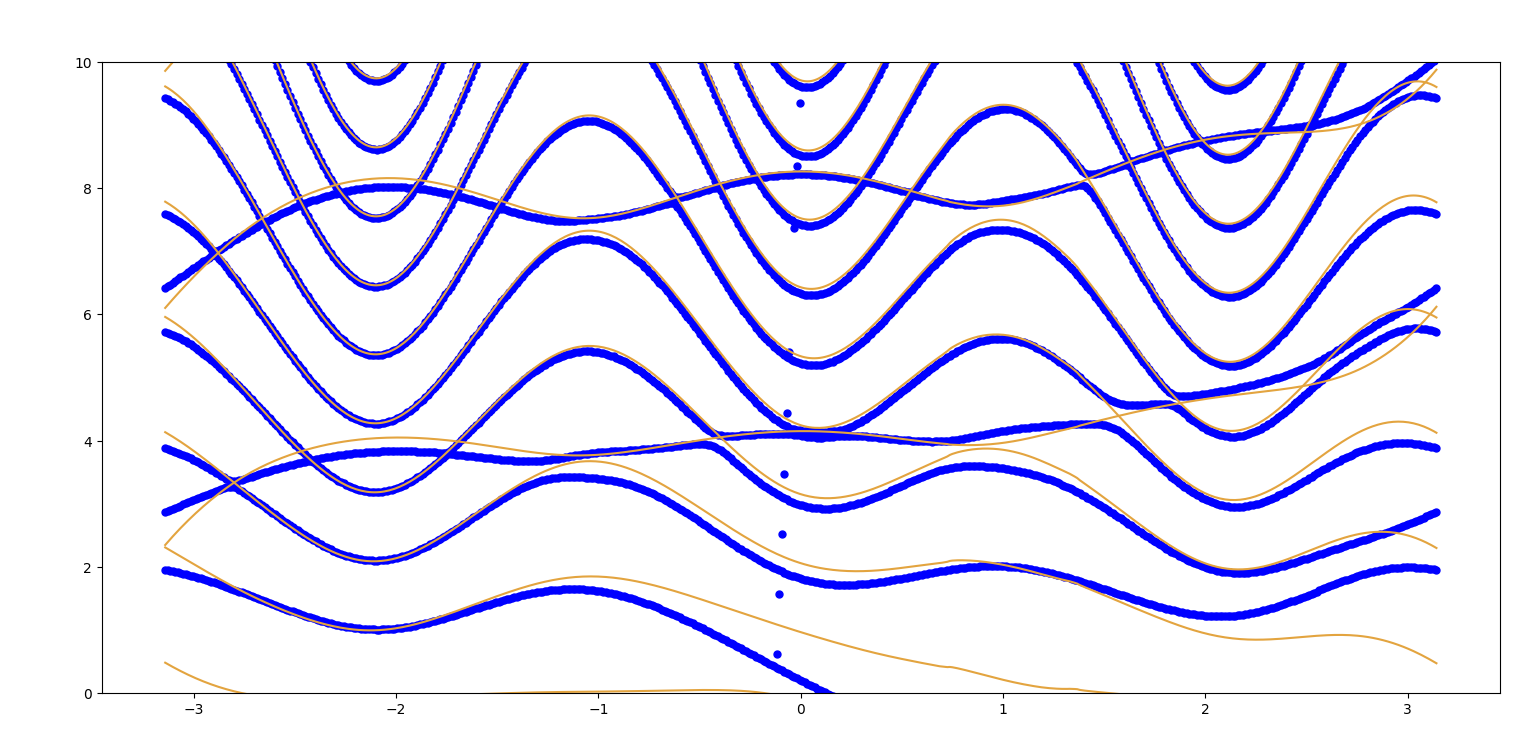}
  \caption{The spectrum of $\widetilde{H}$ at $E_F = 1.5$, right figure, and $E_F = -1.5$, left figure, in blue,
  compared to \eqref{approxSpectrum} in orange.}
  \label{Approx Spectrum}
\end{figure}
Numerics for this are discussed in Section~\ref{Numerics}.

\section{Fourier Transform of the Operator $\widetilde{H}$}
Throughout this section we suppress the $k_y$ dependence.
Let $H = \sum_{x} H_x$ be given as in Definition \ref{Htildedef}.
The Fourier transform of $H_0$ is an operator with integral kernel $H_0(k_x, k_x')$ acting on $L^2(S^1, \mathbb{C}^n)$.
For our choice of $H_0$,
it is given by\footnote{In general, we may have \begin{align*}
  H_0(k_x, k_x') &= \sum_{-N' < m < N} V_m e^{i m (k_x - k_x')}
  + \sum_{M' < m < M} A_m e^{i m k_x} e^{-i (m-1) k_x'} + A_m^* e^{i m k_x'} e^{-i (m-1) k_x }
\end{align*} where $\sum_{m} V_m = V$ and $\sum_{m} A_m = A$. Unless $A_m = 0$ for $m \neq 0$ or $m \neq 1$ defining an edge is difficult.}
\begin{align*}
  H_0(k_x, k_x') &= V
  + A e^{i k_x}  + A^* e^{- i  k_x'},
\end{align*}
for some $n \times n$ Hermitian matrix $V$, and arbitrary $n \times n$ matrix $A$.
We see that the Fourier transform of $H$ is 
\begin{gather*}
  \mathcal{F} H \mathcal{F}^* 
  = \frac{1}{2\pi} \sum_{m} e^{-i k_x m} H_0(k_x, k_x') e^{i k_x' m} 
  = \delta(k_x - k_x') H_0(k_x, k_x').
\end{gather*}
So we get that $H(k_x) = H_0(k_x, k_x)$.
Now, the Fourier transform of $\widetilde{H}$ is 
\begin{align*}
  \widetilde{H} = \frac{1}{2\pi} \sum_{x} x e^{-i k_x x} H_0(k_x, k_x') e^{i k_x' x}
  &= -i \delta'(k_x - k_x') H_0(k_x, k_x') \\
  &= i \delta(k_x - k_x') \partial_{k_x'} H_0(k_x, k_x') \\
  &= i H(k_x) \frac{d}{dk_x} + i \partial_{k_x'} H_0(k_x, k_x')|_{k_x = k_x'}.
\end{align*}
We are using the same notation for $\widetilde{H}$ and it's Fourier transform, we think that the chance of confusion is small. Explicitly, this is the operator 
\begin{gather}\label{Htilde}
  \widetilde{H} = i H(k_x) \frac{d}{dk_x} 
  + A^* e^{-i k_x}.
\end{gather}
By Assumption~\ref{assumption:bulk}, $H$ is nondegenerate and standard ODE theory \cite{Weidmann_1987} gives that $\widetilde{H}$ is self-adjoint on the definition domain $H^1(S^1, \mathbb{C}^n)$.
\begin{lemma}
  The spectrum of $\widetilde{H}$ is discrete, and we may find eigenfunctions realizing these eigenvalues.
\end{lemma}
\begin{proof}
  Since $H$ is invertible for all $k_x$ this is again standard ODE theory, see e.g. Section 7 of \cite{Weidmann_1987}.
\end{proof}
Recall that $\widetilde{H}_{\mathrm{edge}}$ is a restriction of $\widetilde{H}$ to $\ell^2(\mathbb{N} \times \mathbb{Z}, \mathbb{C}^n)$. We also define $\widetilde{H}_{\text{edge}}^{-}$ to be the restriction of $\widetilde{H}$ to  $\ell^2(\Z_{\leq 0} \times \Z, \C^n)$. Since $\widetilde{H}$ is self-adjoint we have the corresponding decomposition of spectra.
\begin{lemma}\label{decomposition}
  We have a decomposition $\sigma(\widetilde{H}(k_y)) = \sigma(\widetilde{H}_{\text{edge}}(k_y)) \cup \sigma(\widetilde{H}_{\text{edge}}^{-}(k_y))$.
\end{lemma}

\section{Large Energy Asymptotics}
We primarily focus on the Fourier representation of $\widetilde{H}$.
We first establish an asymptotic formula for the spectrum of $\widetilde{H}(k_y)$.
We then establish an asymptotic formula for the eigenfunctions of $\widetilde{H}(k_y)$ away from degeneracy.
Finally, we distinguish which branches of the spectrum belong to $\widetilde{H}_{\text{edge}}$.

For simplicity, we write $M(k_x, k_y)$ for $A^*(k_y) e^{-i k_x}$.
The eigenvalue equation for $\widetilde{H}(k_y)$ is
\begin{gather}
  i H(k_x, k_y) \frac{d}{dk_x} \psi(k_x, k_y, E) + M(k_x, k_y) \psi(k_x, k_y, E) = E \psi(k_x, k_y, E).
\end{gather}
Since $H$ is nonsingular, we may write this in the form 
\begin{gather}
  i \frac{d}{dk_x} \psi(k_x, k_y, E) = (E H^{-1}(k_x, k_y) - H^{-1}(k_x, k_y) M(k_x, k_y)) \psi(k_x, k_y, E)
\end{gather}
which has solution 
\begin{gather}
  \psi(k_x, k_y, E) = U_E(k_x, k_y) \psi(0, k_y, E)
\end{gather}
for 
\begin{gather}
  U_E(k_x, k_y) \coloneq \mathcal{T} \exp \left(-i \int_{0}^{k_x} E H^{-1}(t, k_y)- H^{-1}(t, k_y) M(t, k_y) dt \right).
\end{gather}
The definition domain of $\widetilde{H}(k_y)$ requires that $\psi(0, k_y, E) = \psi(2 \pi, k_y , E)$. A value of $E$ is then in the spectrum of $\widetilde{H}(k_y)$ if and only if $U_E(2\pi, k_y)$ has an eigenvalue of $1$.
If the corresponding eigenvector is $\psi$, we may write our solution as 
\begin{gather}\label{solutionEq}
  \psi(k_x, k_y, E) = U_E(k_x, k_y) \psi.
\end{gather}

We need a few technical lemmas to analyze the spectrum of $U_E(2\pi, k_y)$. For the next lemma, we recall that $H(k_x, k_y)$ has
  eigenvalues $E_j(k_x, k_y)$ and corresponding eigenvectors $\varphi_j(k_x, k_y)$. We denote $W(k_x, k_y)$ the parallel transport from $\varphi_j(0, k_y)$ to $\varphi_j(k_x, k_y)$, generated by Kato's parallel transport genrator,
  $$
  A(s) = \sum_{j = 1}^{n} \frac{i}{2} [\dot{P}_{j}(s,k_y), P_{j}(s, k_y)],
  $$
   where $P_j(k_x, k_y)$ the projection on $\varphi_j(k_x, k_y)$.
\begin{lemma} \label{adiabatic}
  Let $E > 0$\footnote{For $E < 0$, there is a negative sign on the terms $G_Z$ and $A$}. 
Let 
  \begin{gather}
    H_W(k_x, k_y) 
  = \sum_{j = 1}^{n} P_j(0, k_y) \int_{0}^{k_x} (E E_j(t, k_y)^{-1} -  E_j(t, k_y)^{-1} \braket{\varphi_j(t,k_y)|M(t,k_y)|\varphi_j(t,k_y)})dt
  \end{gather}
  and define 
  \begin{gather}
    \widetilde{U}_{E}(k_x, k_y) = W(k_x, k_y) e^{-i H_W(k_x, k_y)}.
  \end{gather}
  Then 
  \begin{gather}
    \norm{U_{E}(k_x, k_y) - \widetilde{U}_{E}(k_x, k_y)} \leq C E^{-1},
  \end{gather}
  for $C$ independent of $E, k_x, k_y$.
\end{lemma}

Denoting $W(k_x, k_y) \varphi_j(0, k_y) = e^{i \theta_j(k_x, k_y)} \varphi_j(k_x, k_y)$ we can write explicitly,
  \begin{gather}\label{adiabaticEvo}
    \widetilde{U}_E(k_x, k_y) 
    = \sum_{j = 1}^{n}  e^{{i \theta_{j}(k_x, k_y) - i E \int_{0}^{k_x} E_j(t, k_y)^{-1} (1 - \braket{\varphi_{j}(t,k_y)|M(t,k_y)|\varphi_{j}(t,k_y)})}}
    \ket{\varphi_j(k_x, k_y)} \bra{\varphi_{j}(0, k_y)}.
  \end{gather}

\begin{proof}
  This is essentially Theorem 5 of \cite{AtoZ}.
  The method of proof needs to be slightly modified.
  First note that we have uniform bounds on the norm of $U_{E}$, $U_E^{-1}$, $\widetilde{U}_E$, and $\widetilde{U}_E^{-1}$
  uniform in $k_x,k_y$ from compactness and uniform in $E$ since the non-self-adjoint part of the evolution equation is independent of $E$.
  Following the proof of their Lemma 1, set $U_1 = \widetilde{U}_{E}$ and $U_2 = U_E$, and see
  \begin{gather*}
    \int_{0}^{t} ds~U_1(t) \frac{d}{ds} [U_1(s)^{\dagger} U_2(s)]
    = U_1(t) U_1(t)^{\dagger} U_2 - U_1.
  \end{gather*}
  Let $\frac{d}{dt}U_i = H_i U_i$. Then since we have uniform control on $U_1$, $U_2$ we get the bound 
  \begin{gather}\label{boundadiabatic}
    \norm{U_1 U_1^{\dagger} U_2 - U_1} \leq C \norm{S_{21}}_{\infty}
  \end{gather}
  with 
  \begin{gather*}
    S_{21}(t) = \int_{0}^{t} ds~ [H_1^{\dagger}(s) - H_2(s)].
  \end{gather*}
  Now, going to their proof of Theorem 5, provided 
  \begin{gather}
    H_1^{\dagger}(s) = E H_0(s) + G_{Z}(s) + A(s)
  \end{gather}
  for
  \begin{gather*}
    G_{Z}(s) = \sum_{\ell = 1}^{n} P_j(k_x, k_y) H^{-1}(t) M(t) P_j(k_x, k_y) = \sum_{\ell = 1}^{n} \ket{\varphi_j(t)} \braket{\varphi_j(t) | H^{-1}(t) M(t) | \varphi_j(t)} \bra{\varphi_j(t)}
  \end{gather*}
  we can estimate the action the same way,
  because in our system the term with $E$ is self adjoint.
  To calculate $U_1 U_1^{\dagger}$ write 
  \begin{gather*}
    U_1 = \sum_{j} e^{-i a_j + b_j} \ket{\varphi_j(t)} \bra{\varphi_j(0)}
  \end{gather*}
  for $a_j + i b_j = \int \overline{\lambda_j} dt - \theta_j$, where $\overline{\lambda_j}$ the $j$th eigenvalue of $H_1^{\dagger}$.
  We may compute 
  \begin{align*}
    U_1(t) U_1(t)^{\dagger}
    &= \left(\sum_{j} e^{-i a_j(t) + b_j(t)} \ket{\varphi_j(t)} \bra{\varphi_j(0)} \right)
      \left( \sum_{k} e^{i a_k(t) + b_k(t)} \ket{\varphi_k(0)} \bra{\varphi_k(t)} \right) \\
      &= \sum_{j} e^{2 b_j(t)} \ket{\varphi_j(t)} \bra{\varphi_j(t)}.
  \end{align*}
  Factoring out the terms in the bound \eqref{boundadiabatic} we thus see 
  \begin{gather*}
    \norm{U_2 - \widetilde{U_1}} \leq C E^{-1}
  \end{gather*}
  for 
  \begin{gather*}
    \widetilde{U_1} = \sum_{j} e^{-i a_j(t) - b_j(t)} \ket{\varphi_j(t)} \bra{\varphi_j(0)}.
  \end{gather*}
  Note this has the effect of reversing the conjugation on $\lambda_j$, so
  \begin{gather*}
    \widetilde{U}_1 = \sum_{j} e^{-i \int \lambda_j dt + i \theta_j}\ket{\varphi_j(t)} \bra{\varphi_j(0)}
  \end{gather*}
  as in \eqref{adiabaticEvo}.
\end{proof}
  In particular, this gives a formula for the eigenvalues of $U_E(2\pi, k_y)$ as 
\begin{gather}\label{eigenvalEq}
  \exp\left(i \theta_j(2 \pi, k_y) -  i \int_{0}^{2 \pi} E E_j^{-1}(t, k_y) - E_j(t, k_y)^{-1} \braket{\varphi_j(t, k_y) | M(t,k_y) | \varphi_j(t, k_y)}dt\right)
\end{gather}
  where $\theta_j(2 \pi, k_y) \equiv \theta_j(k_y)$ is the Berry phase corresponding to the $j$th band.
\begin{lemma}
  Suppose that $\widetilde{U}_{E}(2\pi, k_y)$ have an eigenvalue of one and $E > 0$. Then 
  \begin{gather}
    E = E(m, j, k_y) \equiv \frac{2 \pi m + \theta_j(k_y) + \int_{0}^{2\pi} E_j(k_x, k_y)^{-1} \braket{\varphi_j(k_x, k_y)|M(k_x, k_y)|\varphi_j(k_x, k_y)} dk_x}{\int_{0}^{2\pi} E_j(k_x, k_y)^{-1} dk_x}
  \end{gather}
  for $m \in \Z$, $m$ such that this $E$ is positive, $1 \leq j < n$.
\end{lemma}
\begin{proof}
  With \eqref{eigenvalEq} for the eigenvalues, we set the argument of the exponential to $2 \pi m$ and solve for $E$.
\end{proof}
\begin{lemma}
  Let $\widetilde{U}_{E}(k_x, k_y)$ be given as above.
  Then $\widetilde{U}_{E}(2\pi, k_y)$ is normal.
\end{lemma}
\begin{proof}
  Since $H(0, k_y) = H(2\pi, k_y)$, projections onto the $\varphi_j(0, k_y)$ eigenspace and $\varphi_j(2\pi, k_y)$ are the same.
  Hence $\varphi_j(0, k_y)$ is a unitary frame which diagonalizes $\widetilde{U}_{E}(2\pi, k_y)$.
\end{proof}
\begin{lemma}\label{eigenvalClose}
  Let $U_E(2\pi, k_y)$ and $\widetilde{U}_E(2\pi, k_y)$ be given as above.
  Then for every eigenvalue $\lambda$ of $U_E(2\pi, k_y)$, there is an eigenvalue $\widetilde{\lambda}$ of 
  $\widetilde{U}(2\pi, k_y)$ such that $|\lambda - \widetilde{\lambda}| < C E^{-1}$,
  where $C$ is from Lemma \ref{adiabatic}.
\end{lemma}
\begin{proof}
  This is Corollary 6.3.4 from \cite{HornJohnson}.
  The fact $\widetilde{U}_E(2\pi, k_y)$ is normal is important.
\end{proof}

We now establish an asymptotic formula for the spectrum.
\begin{lemma}\label{spectrumForm}
  Let $E > 0$  be such that $U_E(2\pi, k_y)$ has an eigenvalue of one,
  that is $E \in \sigma(\widetilde{H}(k_y))$.
  Then we may find $m,j$ such that 
  \begin{gather}
    |E - E(m, j, k_y)| < C |E|^{-1}
  \end{gather}
  for $C$ independent of $E, k_x, k_y$.
\end{lemma}
\begin{proof}
  We use $C$ as a generic constant here.
  For fixed $k_y$, the eigenvalues of $\widetilde{U}_E(2\pi, k_y)$ may be written as 
  \begin{gather*}
    f_j(E) = e^{-i c_j E + d_j}
  \end{gather*}
  where $c_j = \int E_j^{-1} dk_x$.
  These $f_j(E)$ are holomorphic in $E$,
  and $f_j(E) = 1$ if and only if  $E = E(m, j, k_y)$.
  For our given $E$, we know by Lemma \ref{eigenvalClose} that 
  \begin{gather*}
    |1 - f_j(E)| < C E^{-1} = \epsilon.
  \end{gather*}
  Since $f_j$ is holomorphic and periodic
  and takes on values $1$ only at $E(m, j, k_y)$ we must have 
  \begin{gather*}
    |E - E(m,j,k_y)| < C \epsilon
  \end{gather*}
  for some $m,j$ and some $C$ depending on the derivative of $f$.
  We have uniform bounds on $f_j' = |c_j'| |e^{-i c_j E + d_j}|$
  as $|c_j'| = |\int E_j^{-1} dk_x|$ is bounded and bounded away from zero by Assumption \ref{assumption:bulk},
  so we may find some $C$ independent of $k_x, k_y$ so that this holds.
\end{proof}

Now we find an asymptotic expression for the eigenfunctions. 
There are complications here if $U_{E}(2 \pi, k_y)$ is degenerate,
so for now we assume this is not the case. 
\begin{lemma}\label{functionsClose}
  Let $k_y$ be such that the 
  $E(m,j, k_y)$ are distinct for different $m, j$.
  Let $E \in \sigma(\widetilde{H}(k_y))$ and write $\psi(k_x, k_y, E)$ as in \eqref{solutionEq},
  and set 
  \begin{gather}
    \widetilde{\psi}(k_x, k_y, E(m,j, k_y))
    \coloneq \widetilde{U}_{E(m, j, k_y)}(k_x, k_y) \varphi_j(0, k_y).
  \end{gather}
  Then for $E$ large $U_E(2 \pi, k_y)$ is nondegenerate, and 
  if $|E - E(m, j, k_y)| < C E^{-1}$ then 
  \begin{gather}
    \norm{\psi(k_x, k_y, E) - \widetilde{\psi}(k_x, k_y, E(m,j, k_y))}_{L^{\infty}} 
    \leq C(k_y) E^{-1}
  \end{gather}
  for $C(k_y)$ depending on $k_y$.
\end{lemma}
\begin{proof}
  We first need to establish $U_E(2 \pi, k_y)$ is nondegenerate for $E$ large enough.
  Using Lemma \ref{eigenvalClose} this follows since $\widetilde{U}_E(2 \pi, k_y)$ is nondegenerate at the 
  $E(m, j, k_y)$ by assumption.
  Then 
  \begin{align}
    \norm{\psi(k_x, k_y, E) - \widetilde{\psi}(k_x, k_y, E(m,j,k_y))}
    &= \norm{U_E(k_x, k_y) \psi(0) - U_{E(m,j,k_y)}(k_x, k_y) \varphi_j(0)}.
  \end{align}
  First, see that $\varphi_j$ is an eigenvector of $\widetilde{U}_{E}(2\pi, k_y)$ for all $E$ and $k_y$.
  By Lemma \ref{adiabatic} $\widetilde{U}_E(2\pi, k_y)$ is distance $C E^{-1}$ to 
  $U_E(2\pi, k_y)$.
  Then by standard matrix perturbation theory (\cite{kato2013perturbation} Section 2),
  since $\widetilde{U}_E(2\pi, k_y)$ and $U_E(2\pi, k_y)$ are nondegenerate,
  we may choose an eigenvector $\psi$ of $U_E(2\pi, k_y)$ such that 
  \begin{gather}
    \norm{\psi - \varphi_j(0)} < C(k_y) E^{-1}
  \end{gather}
  where $C$ depends inversely on the size of the eigenvalue gap of $U_E(2\pi, k_y)$.
  Next, see 
  \begin{gather}
    \norm{U_E(k_x, k_y) \psi(0) - U_{E(m,j,k_y)}(k_x, k_y) \varphi_j(0)} \notag \\
    = \norm{U_E(k_x, k_y) \psi(0) - U_E(k_x, k_y) \varphi_j(0) + U_E(k_x, k_y) \varphi_j(0) - \widetilde{U}_{E(m,j,k_y)(k_x, k_y)} \varphi_j(0)}
  \end{gather}
  and 
  \begin{align}
    \norm{U_E(k_x, k_y) \psi(0) - U_E(k_x, k_y) \varphi_j(0)}
    \leq C \norm{\psi(0) - \varphi_j(0)}
    \leq C(k_y) E^{-1}
  \end{align}
  and
  \begin{align}\label{triangle}
    \norm{U_E(k_x, k_y) \varphi_j(0) - \widetilde{U}_{E(m, j, k_y)} \varphi_j(0)}
    &= \norm{U_E(k_x, k_y) - \widetilde{U}_E(k_x, k_y) + \widetilde{U}_E(k_x, k_y) - \widetilde{U}_{E(m, j, k_y)}} \notag \\
    &\leq C E^{-1}
  \end{align}
  where in \eqref{triangle} we have used Lemma \ref{adiabatic} and that the eigenvalues of $\widetilde{U}_E$ are Lipshitz in $E$.
  Hence we get the result.
\end{proof}

We now establish which eigenvalues belong to $\widetilde{H}_{\text{edge}}$.
By Lemma \ref{decomposition}, to determine if an eigenvalue belongs to the edge Hamiltonian
we need only check its Fourier transform has a nonvanishing positive Fourier mode,
if the energy is non-degenerate this is equivalent to its Fourier transform having no negative Fourier modes. 
We are aiming for

\begin{lemma}\label{seperatingBands}
  Let $k_y$ be such that the hypothesis of Lemma \ref{functionsClose} holds. Let large enough  
  $E$ be close to $E(m, j, k_y)$, in the sense of Lemma~\ref{spectrumForm}, for $E_j < 0$.
  Then $E \in \sigma(\widetilde{H}_{\text{edge}}(k_y))$.
\end{lemma}
We accomplish this in a few steps. In the following lemmas, $k_y$ is fixed, 
satisfying the conditions of Lemma \ref{functionsClose}.
\begin{lemma}\label{O(k)}
  Let $n$ be negative. 
  Then $|\widehat{\widetilde{\psi}}_{E(m,j,k_y)}(n)| < C E(m,j,k_y)^{-k}$ for all $k \in \N$,
  for $C$ independent of $m$.
\end{lemma}
\begin{proof}
  For simplicity, write 
  \begin{gather}
    \widetilde{\psi}_{E(m,j,k_y)}(k_x, k_y) = e^{-i E c_j(k_x, k_y)} g_j(k_x, k_y)
  \end{gather}
  for $c_j,g_j$ independent of $E$ and $c_j(k_x, k_y) = \int E_j^{-1} dk_x$.
  The Fourier coefficients of $\widetilde{\psi}$ are 
  \begin{gather}
    \widehat{\widetilde{\psi}}(n)
    = \int_{0}^{2 \pi} e^{-i n k_x}  e^{-i E c(k_x, k_y)} g(k_x, k_y) dk_x.
  \end{gather}
  Note 
  \begin{gather*}
    e^{-i n k_x - i E c(k_x, k_y)}
    = \frac{i}{n + E c'(k_x, k_y)} \frac{d}{dk_x} e^{-i n k_x - i E c(k_x, k_y)}
  \end{gather*}
  and since $n$ is negative and $c' = E_j^{-1}$ is negative the denominator is nonzero, hence an integration by parts gives 
  \begin{align*}
    \widehat{\widetilde{\psi}}(n) &= \int_{0}^{2\pi} \frac{-i}{n + E c'(k_x, k_y)} e^{i n k_x -i E c(k_x, k_y)} g'(k_x, k_y) dk_x  \\
    &+ \int_{0}^{2\pi} \frac{-i E c''(k_x, k_y)}{(n + E c'(k_x, k_y))^2 }e^{i n k_x -i E c(k_x, k_y)} g(k_x, k_y) dk_x.
  \end{align*}
  Since $n$ is negative, the signs align and $|n + E c'|^{-1} < (E c')^{-1}$.
  Hence we see the bound holds for $k = 1$.
  Repeated integration by parts gives the bound in general.
\end{proof}
\begin{lemma}
  Let $n$ be negative. 
  Then $|\widehat{{\psi}}_{E}(n)| < C E^{-1}$
  for $C$ independent of $E$.
\end{lemma}
\begin{proof}
  Follows from Lemmas \ref{functionsClose} and \ref{O(k)}.
\end{proof}
\begin{lemma}\label{L2}
  The $L^2$ norms $\norm{\psi(k_x, k_y, E)}_{L^2}$ are bounded below 
  away from zero in $E$.
\end{lemma}
\begin{proof}
  First see that $\norm{\widetilde{\psi}_{E(m,j, k_y)} (k_x, k_y)}_{L^2}$ 
  is independent of $E$, since the term with $E$ is a pure phase.
  Hence the result follows by Lemma \ref{functionsClose}.
\end{proof}

We can now prove Lemma \ref{seperatingBands}.
\begin{proof}
  Suppose, for contradiction, that $\psi(k_x,  k_y, E)$ only had negative Fourier coefficients.
  By Lemma \ref{O(k)}, the Fourier coefficients of $\psi(k_x, k_y, E)$ are $O(E^{-1})$.
  Hence we can make the $L^2$ norm of $\psi$ as small as we want.
  This contradicts Lemma \ref{L2}.
\end{proof}
\begin{lemma}
  There are finitely many $k_y$ where the $E(m,j, k_y)$ intersect.
\end{lemma}
\begin{proof}
  Since $H(k_x, k_y)$ is analytic, the $E(m,j,k_y)$ are analytic in $k_y$ and can only intersect at finitely many points,
  or entirely overlap.
  Suppose then that $E(m,j,k_y) = E(k, \ell, k_y)$. Then rearranging,
  \begin{gather*}
    2 \pi m = \left(\frac{\int E_j^{-1} dk_x}{\int E_k^{-1} dk_x} \right) \left(2 \pi n + \theta_k(k_y) + \int E_k^{-1} \braket{\psi_k|A|\psi_k} dk_x \right) 
    - \theta_j(k_y) - \int E_j^{-1} \braket{\psi_j|A|\psi_j} dk_x.
  \end{gather*}
  Evaluate at $k_y = 2 \pi$ and $k_y = 0$ and subtract. By periodicity we are left with
  \begin{gather*}
    \left( \frac{\int E_j^{-1} dk_x}{\int E_k^{-1} dk_x}  \right) (\theta_k(2\pi) - \theta_k(0)) 
    = (\theta_j ( 2\pi) - \theta_j(0)).
  \end{gather*}
  If $E_j$ and $E_k$ have different signs, then $\theta_j$ and $\theta_k$ have winding numbers of opposite sign.
  Otherwise without loss of generality assume $|E_j| < |E_k|$.
  Then $\theta_j(2\pi) - \theta_j(0) > \theta_k(2\pi) - \theta_k(0)$
  and $\theta_j$ and $\theta_k$ have different winding numbers.
  Hence $E(m,j,k_y)$ and $E(k, \ell, k_y)$ cannot be the same.
\end{proof}

With this, we can now prove Theorems \ref{WindingNumber} and \ref{Formula}.

\begin{proof}
  Excise a neighborhood around the points where 
  the $E(m,j,k_y)$ intersect. Then there exists a $C > 0$ such that 
  $|E(m,j,k_y) - E(k, \ell, k_y)| > C$ for all $m,k,j,\ell$.
  As such Lemma \ref{functionsClose} holds for a $C$ independent of $k_y$ for all $k_y$ besides those which have been excised.
  Thus by Lemma \ref{seperatingBands} the asymptotic formula in Theorem \ref{Formula} holds,
  for the $k_y$ not excised. By continuity of the spectrum in $k_y$, it must hold for all $k_y$.
  This proves Theorem \ref{Formula}.
  Since all terms in the asymptotic formula are periodic in $k_y$ besides the Berry phases,
  we compute the winding number as 
  \begin{align*}
    \text{wind}(\sigma(\widetilde{H}_{\text{edge}}(k_y))) 
    &= \sum_{E_j < 0} \text{wind}(\theta_j(k_y)) \\
    &= \frac{1}{2\pi} \sum_{E_j < 0} (\theta_j(2\pi) - \theta_j(0)) \\
    &= \kappa,
  \end{align*}
where $\kappa$ is the Hall conductance, see Section~\ref{sec:setup}.  This proves Theorem \ref{WindingNumber}.
\end{proof}

\section{Numerical Experiments}\label{Numerics}
We illustrate our results by calculating spectra flows of a discretization of $\widetilde{H}$ in 
\eqref{Htilde}.
We denote the spectral Fourier discretization of $\widetilde{H}$ by $\mat{H}$.
The main numerical challenge is finding the correspondence between eigenpairs of the continuous operator 
$\widetilde{H}$ and that of the discretized operator $\mat{H}$.
This is complicated by the fact that $\mat{H}$ has spurious eigenvalues corresponding to discretization errors due to not fully resolving high frequency modes.
To overcome this difficulty, we compute the eigenpairs using parallel transport (cf. Section 7 of \cite{kato2013perturbation}) following the approach of 
\cite{gopal2025highly}.
There, they show that the parallel transport of an eigenpair can be computed by integrating 
\begin{align*}
  \mat{\psi}'(k) &= - (\mat{H}(k) - E(k) \mat{I})^{\dagger} \mat{H}'(k) \mat{\psi}(k) \\
  E'(k) &= \mat{\psi}^*(k) \mat{H}'(k) \mat{\psi}(k)
\end{align*}
where $\mat{I}$ denotes the identity matrix and $\dagger$ denotes the Moore-Penrose pseudoinverse.
We compute an initial condition using the QR algorithm and then integrate the differential equation using a fourth-order Runge-Kutta method.

The $\alpha = 1/3$ Harper-Hofstadter model has two nontrivial Chern numbers, $\pm 1$,
depicted in Figure \ref{Harper}.
As predicted by Theorem \ref{Formula}, with Chern number $1$ we have one family of eigenvalues in the positive spectrum of $\widetilde{H}_{\text{edge}}(k_y)$
and winding number $1$, and two families of eigenvalues in the positive spectrum of $\widetilde{H}_{\text{edge}}(k_y)^{-}$ and winding number $-1$,
seen in Figure \ref{Positive Chern}.
\begin{figure}[h!]
    \centering
  \includegraphics[totalheight=4cm]{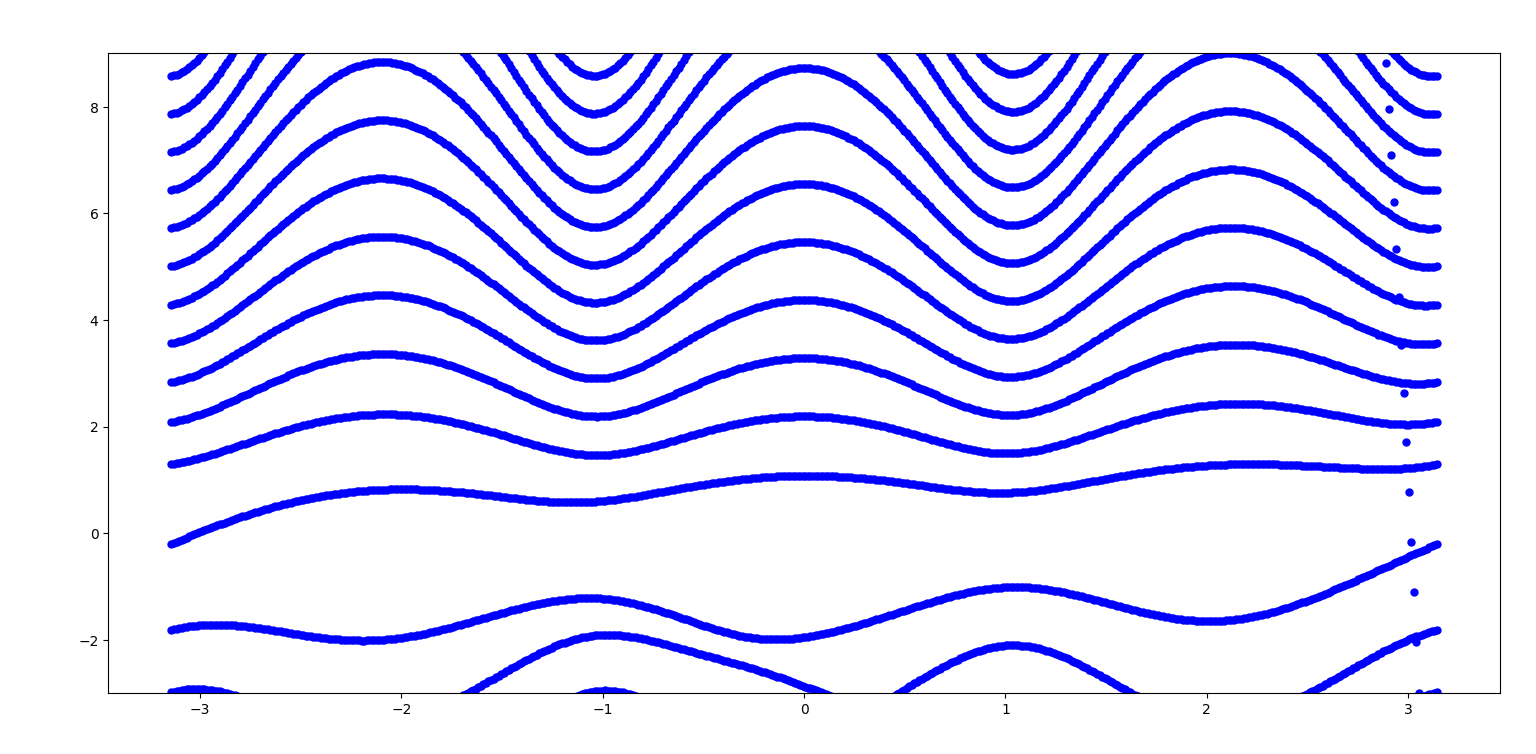}
  \includegraphics[totalheight=4cm]{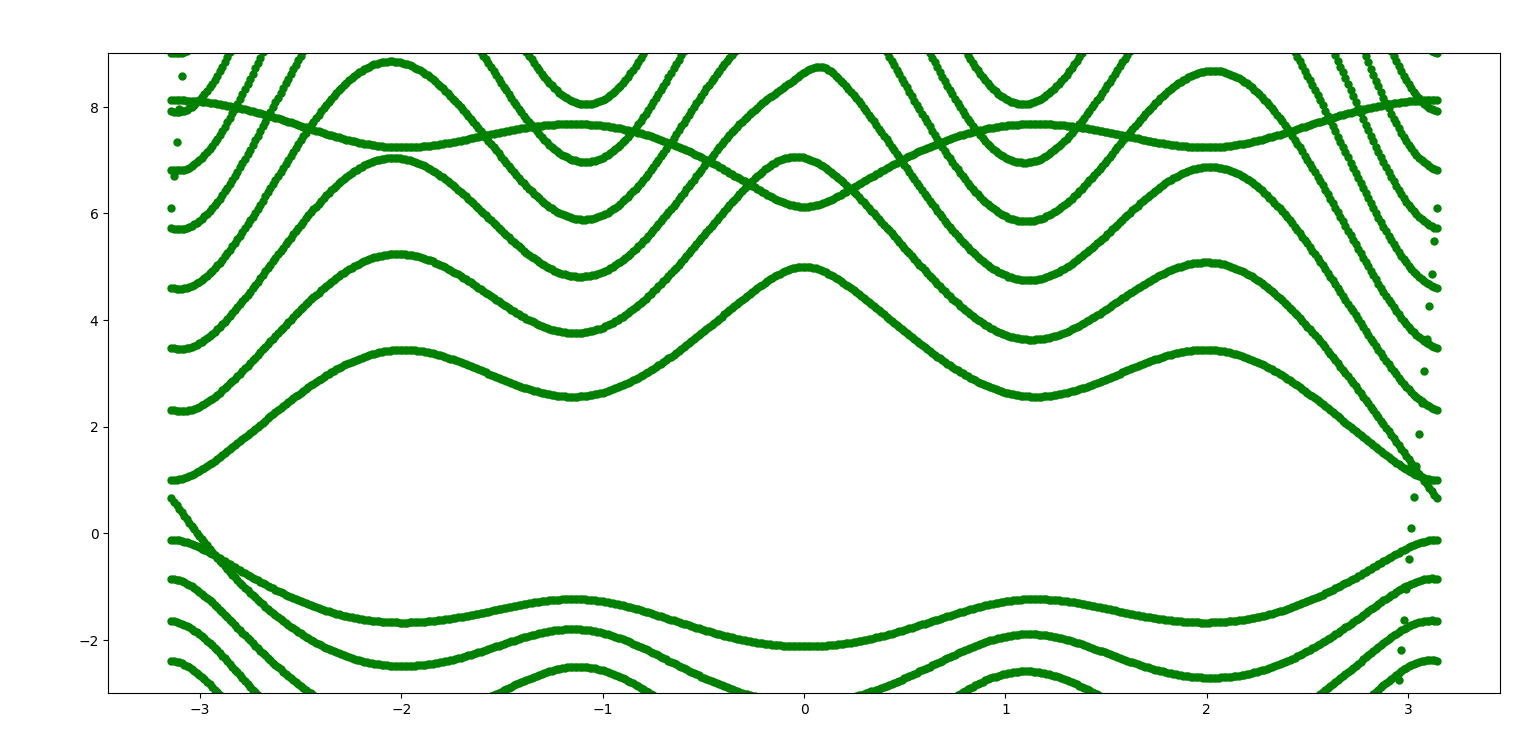}
  \caption{The spectrum of $\widetilde{H}$ at $E_F = 1.5$, sampled at 500 sites. The blue is the spectrum of 
  $\widetilde{H}_{\text{edge}}$, the green the spectrum of $\widetilde{H}_{\text{edge}}^{-}$}
  \label{Positive Chern}
\end{figure}
Analogous statements hold for negative Chern number, seen in Figure \ref{Negative Chern}.
\begin{figure}[h!]
    \centering
  \includegraphics[totalheight=4cm]{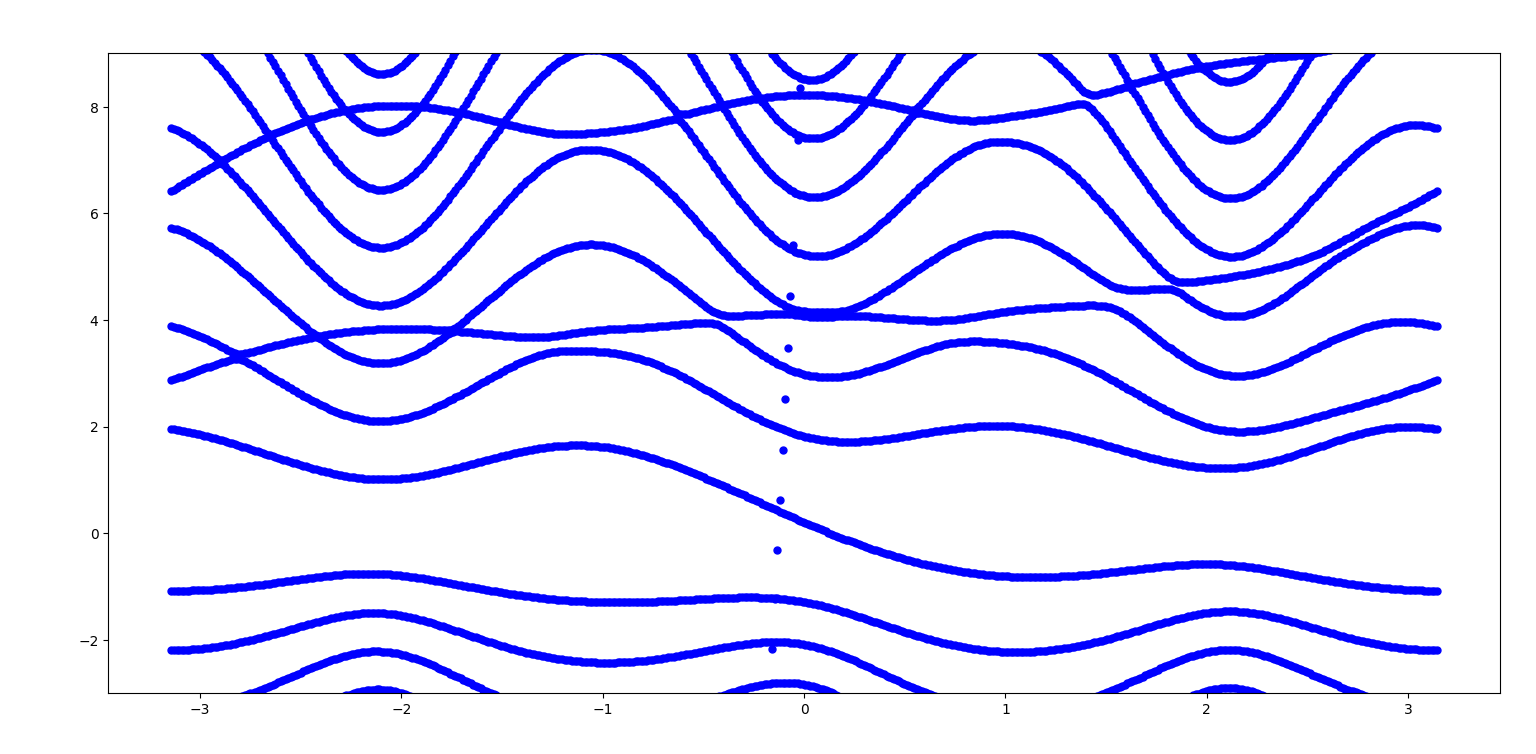}
  \includegraphics[totalheight=4cm]{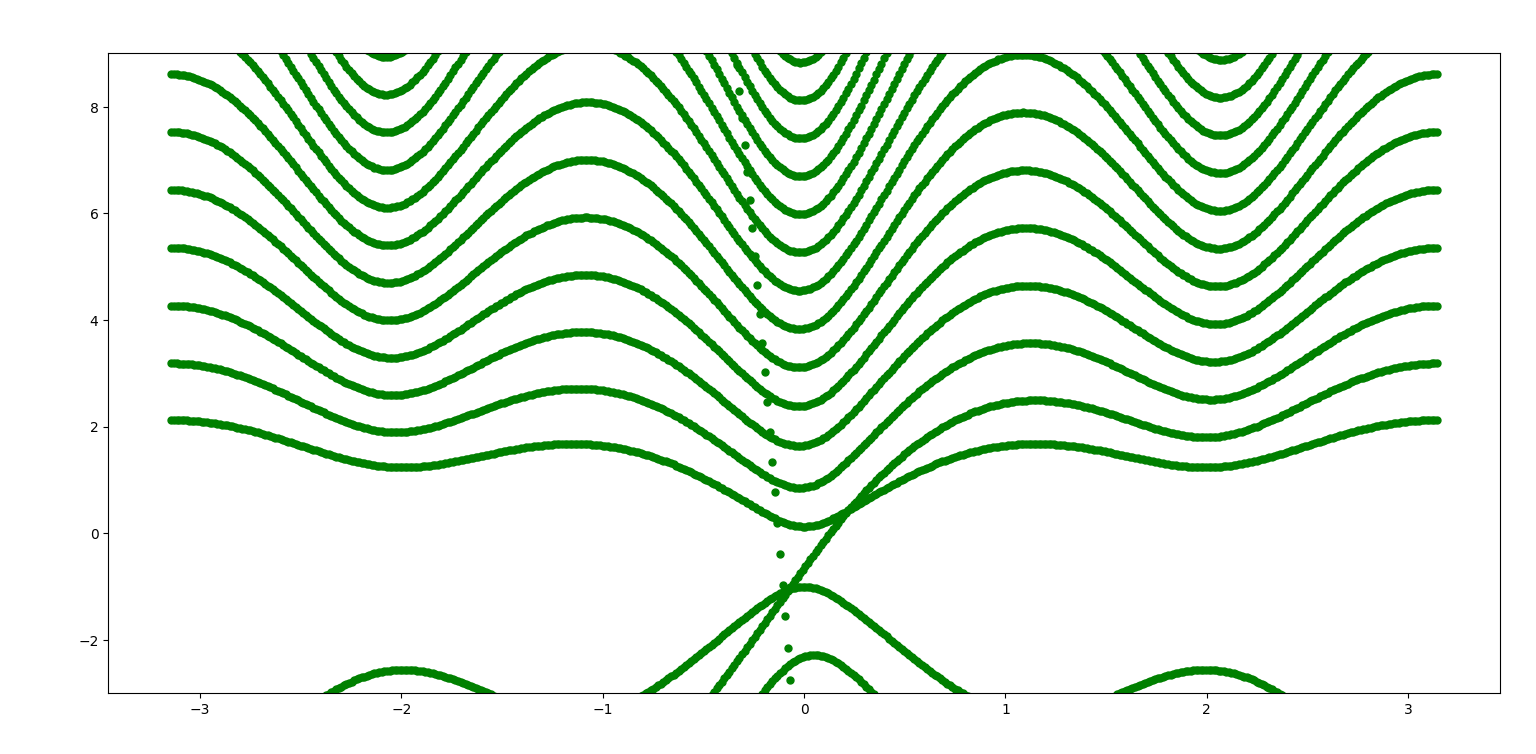}
  \caption{The spectrum of $\widetilde{H}$ at $E_F = -1.5$, sampled at 500 sites. The blue is the spectrum of 
  $\widetilde{H}_{\text{edge}}$, the green the spectrum of $\widetilde{H}_{\text{edge}}^{-}$}
  \label{Negative Chern}
\end{figure}
The dotted lines in the figures are the spurrious eigenvalues,
which Figure \ref{Parallel Transport} demonstrates.
\begin{figure}[h!]
    \centering
  \includegraphics[totalheight=4cm]{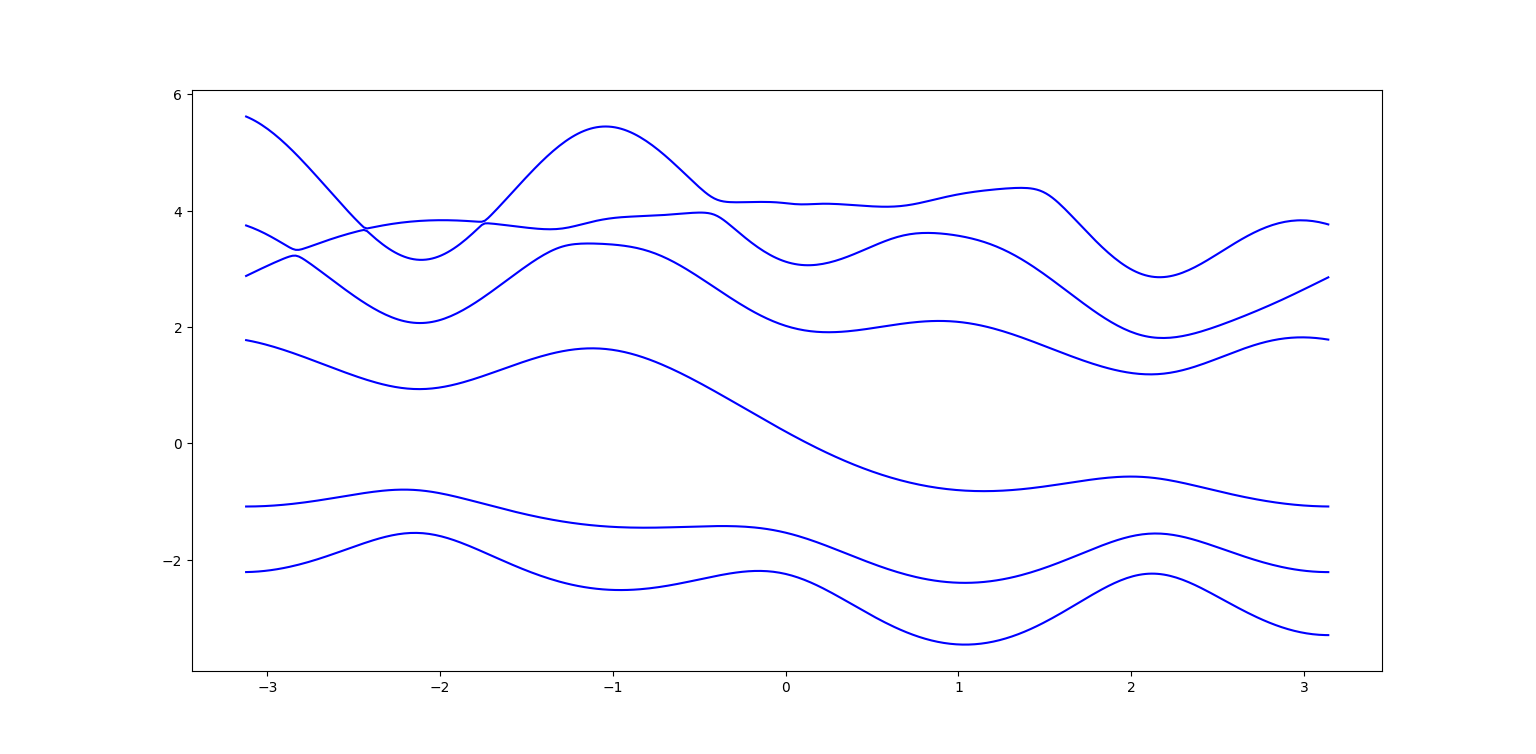}
  \caption{Parallel transport ignoring the dotted lines}
  \label{Parallel Transport}
\end{figure}

\section*{Acknowledgments}
M.F., S.D. and N.S. were supported by the NSF under grant DMS-2407290. M.F. thanks Sven Bachmann for discussions. M.F. is grateful for hospitality at the Research Institute for Mathematical Sciences at Kyoto University.


\bibliography{HeatCurrent}
\bibliographystyle{plain}

\end{document}